\documentclass[letterpaper, 10 pt, conference]{ieeeconf}  

\IEEEoverridecommandlockouts                              

\usepackage{comment}
\usepackage[backend=biber,style=ieee,maxbibnames=6]{biblatex}
\usepackage{amsmath,amssymb,amsfonts}
\usepackage{tabularx}
\PassOptionsToPackage{dvipsnames,svgnames,table}{xcolor}
\usepackage{multirow}
\usepackage{graphicx}
\usepackage{textcomp}
\usepackage{xcolor}
\usepackage{array}
\usepackage{booktabs}
\usepackage{subcaption}
\usepackage[colorlinks=true,allcolors=blue]{hyperref}
\usepackage{placeins}
\usepackage{stfloats}
\usepackage{algorithm}
\usepackage{algpseudocode}
\usepackage{capt-of}

\newtheorem{theorem}{Theorem}
\newtheorem{proposition}{Proposition}
\newtheorem{lemma}{Lemma}
\newtheorem{corollary}{Corollary}
\newtheorem{assumption}{Assumption}
\newtheorem{definition}{Definition}
\newtheorem{remark}{Remark}

\newcommand{\zc}{z^{c}}
\newcommand{\Gbar}{\bar{\Gamma}}
\newcommand{\Exp}{\mathbb{E}}
\newcommand{\Prob}{\mathbb{P}}
\newcommand{\kexc}{\kappa_{\mathrm{exc}}}
\newcommand{\cum}{\operatorname{cum}}
\newcommand{\sgn}{\operatorname{sgn}}
\newcommand{\Gmag}{\mathcal{G}_{\mathrm{mag}}}
\newcommand{\Gsgn}{\mathcal{G}_{\mathrm{sgn}}}

\newcommand{\Cov}{\operatorname{Cov}}
\newcommand{\tr}{\operatorname{tr}}
\newcommand{\xs}{\hat{x}^{s}}
\newcommand{\xa}{\hat{x}^{a}}
\newcommand{\Afil}{A_{\mathrm{f}}}
\newcommand{\Lyap}{\mathcal{L}}
\newcommand{\vsg}{\sigma}
\newcommand{\rlag}{r}                      
\newcommand{\Mset}{\mathcal{M}}            
\newcommand{\fcont}{f}                     
\newcommand{\Neff}{N_{\mathrm{eff}}}

\begin{document}

\title{When Stealth Requires Memory: Budgeted Attack Scheduling under a Whiteness Constraint}


\author{Qazi Mairaj ud din, Sidra Ghayour Bhatti, and Qadeer Ahmed%
\thanks{The authors are with the Center for Automotive Research, The Ohio
State University, Columbus, OH, USA (e-mail: mairajuddin.1@osu.edu;
bhatti.39@osu.edu; ahmed.358@osu.edu). Implementation available at: \url{https://github.com/OSU-CAR-MSL/innov-white-stealthy-attack}}}

\maketitle

\begin{abstract}
We pose stealthy attack scheduling on a sensor-to-estimator link under
a resource constraint with a budget $\Gbar$ on the fraction of corrupted transmissions and a model-free whiteness constraint on the received innovations. For a discrete-time linear plant with non-Gaussian
noise, the worst attack, innovation sign flip, preserves the innovation magnitude and is exactly stealthy against every magnitude-measurable detector, the damage-optimal schedule being a memoryless threshold; but the tail
concentration that maximizes damage also manufactures serial correlation, where an innovation-whiteness monitor gains power. We dualize the whiteness constraint and show the optimum is a threshold rule on \emph{corrected} innovation energy using memory of recent magnitudes and the previous decision. We further trace that correction to a second degree of freedom and set the damage by the \emph{location} of the firing set in magnitude space and set exposure by the \emph{boundary density} of its run structure. We realize it as a hysteresis set by one split-conformal order statistic without any plant model, using one counter and two comparisons per step. It keeps $80$--$96\%$ of the memoryless damage at $2.3$ to $71$ times lower lag-one whiteness power, validated on a real truck CAN record.

\end{abstract}


\section{Introduction}
Stealthy False Data Injection (FDI) on the
sensor-to-estimator link degrades remote state estimation while evading
the residual-based monitors that guard it \cite{zhou2024cybersecurity}.
A practical adversary is resource constrained, its bandwidth and energy
limits expressed as a budget $\Gbar\in(0,1)$ on the fraction of
transmissions it may corrupt, which motivates \emph{event-based}
scheduling: the adversary acts only at instants it judges informative
\cite{han2015stochastic,zhang2015optimal,Guo2023,Wang2025}.

Two parametric ingredients recur in this literature: the firing
threshold inverts a Gaussian tail to meet the budget, and stealth is
certified by requiring the corrupted innovation to retain its nominal
covariance. Neither survives departures from Gaussianity, and real CPS
residuals are routinely non-Gaussian, nonlinearities, saturation, mode
switches and packet drops all generating occasional large innovations
\cite{zhao2021event}. A distribution-free treatment
\cite{din2026distributionfreebudgetedstealthyattack} replaces both
ingredients with three facts that this paper takes as its starting point  and restates in Sec.~\ref{sec:prelim}. The worst-case action is
the \emph{sign flip}, which preserves the innovation magnitude on every
sample path and is therefore \emph{exactly} stealthy against every
detector reading magnitudes alone, for any firing rule and any
innovation law. The induced degradation collapses into one estimable
scalar, the energy capture $\Psi$. And $\Psi$ is maximized by a
\emph{memoryless} threshold on $|z_{k}|$ whose cut point is the same
order statistic that delivers a finite-sample budget guarantee.

That certificate covers detectors measurable with respect to the magnitude process and says nothing about those reading the correlation signature: the innovation-whiteness tests standard in fault detection
\cite{mehra1971innovations,ljung1978measure}. If the innovation is \emph{conditionally sign symmetric} the sign flip preserves
the law of the entire received process and the whiteness tests are
defeated too; that condition holds automatically under Gaussian noise
and fails only off Gaussianity, the obstruction being a fourth cumulant.
When it fails, exposure is \emph{coupled to damage}: the damage-optimal
schedule fires on the largest $|z_{k}|$, and it is in the tail that the
sign--magnitude dependence lives. An adversary that maximizes
damage thereby exposes itself to a whiteness monitor, and one
that hides gives up the damage. Whether that trade-off can be
broken, and at what price, is the question answered here.

We pose budgeted scheduling with an explicit constraint on the
autocovariance of the received sequence and characterize the optimum.
Memory is a consequence of the constraint, and it becomes necessary exactly when conditional sign symmetry fails. Memory has been used before to strengthen innovation-based attacks
\cite{li2019optimal,shang2021optimal,ren2023optimal}: there it increases
damage under a stealth certificate that presumes the innovation law and
the plant matrices, and is imposed on the marginal law or on whiteness
\emph{within} an assumed detection window, which leaves the received
autocovariance unconstrained at the remaining lags. Such attacks are by
construction moving-average or autoregressive in the transmitted stream,
so a Ljung--Box monitor whose span reaches that window exposes them even
under Gaussian noise, the regime in which the sign flip is provably
invisible (Lemma~\ref{lem:boundary}). Furthermore, neither is that line budgeted: it corrupts every transmission, whereas the schedule here meets an explicit rate $\Gbar$ and is itself the object of the design. \emph{Specifically, our contributions are as follows.}
\begin{enumerate}
\item We formulate budgeted scheduling under an exposure constraint on
  the received autocovariance (Sec.~\ref{sec:statement}) and identify
  the mechanism governing it: damage and exposure are distinct
  functionals of a schedule, the \emph{location} of its firing set in
  magnitude space and the \emph{boundary density} of its run structure,
  which the memoryless class, having one degree of freedom, cannot set
  independently (Prop.~\ref{prop:decouple}). Unlike the divergence
  constraints of \cite{guo2018worst,shang2021worst,zhou2022optimal},
  which presume the innovation law, ours is estimated on the stream the
  adversary transmits. The same decomposition accounts for the failure of score randomization, and of off-the-shelf windowed scores (Rem.~\ref{rem:dither}).
\item We characterize the constrained optimum as a threshold on a score
  corrected by past magnitudes and past decisions
  (Thm.~\ref{thm:structure}), and prove it reduces to the memoryless
  benchmark if and only if the innovation is conditionally sign symmetric, which is exactly the Gaussian case (Cor.~\ref{cor:memory}). 
\item We realize that structure causally as a two-threshold rule
  (Sec.~\ref{sec:algorithm}) costing one counter and two comparisons per
  step with no plant matrix, and identify split-conformal calibration
  guarantees for the schedule with memory (Props.~\ref{prop:startrate}--\ref{prop:neff}). The proposed scheduler retains $80$--$96\%$ of the memoryless damage at $2.3$ to $71$ times lower monitor power, and cuts lag-one exposure by up to $4.4\times$ on a real truck CAN record (Sec.~\ref{sec:results}).
\end{enumerate}

\emph{Notation.} $\mathbf1\{\cdot\}$ is the indicator, $\sgn(\cdot)$ the
sign, $\sigma(\cdot)$ the generated $\sigma$-algebra and
$\overset{d}{=}$ equality in law. $\kappa_{r}(X)$ is the $r$-th cumulant
of $X$, $\cum(\cdot)$ the joint cumulant and $\kexc(X)$ the excess
kurtosis. With $\rho(\cdot)$ the spectral radius, $\Lyap(A,W)$ denotes
for $\rho(A)<1$ the unique solution of $X=AXA^{\top}+W$; subscripted,
$\rho_{\ell}$ is the lag-$\ell$ autocorrelation of the firing indicator.
Order statistics of $s_{1},\dots,s_{N}$ are written
$s_{(1)}\le\dots\le s_{(N)}$.

\section{Preliminaries and Problem Formulation}
\label{sec:prelim}

\subsection{System and Threat Model}
\label{sec:sysmodel}
Consider the discrete-time linear plant
\begin{equation}
x_{k+1}=Ax_{k}+w_{k},\qquad y_{k}=Cx_{k}+v_{k},
\label{eq:plant}
\end{equation}
with $x_{k}\in\mathbb{R}^{n}$, $y_{k}\in\mathbb{R}^{m}$, and mutually
independent zero-mean i.i.d.\ noises $w_{k},v_{k}$ of covariance
$Q\succ0$, $R\succ0$. \emph{Neither is assumed Gaussian.} With $(A,C)$
detectable and $(A,Q^{1/2})$ stabilizable the algebraic Riccati equation
has a unique stabilizing solution $\bar P$; put $S:=C\bar PC^{\top}+R$,
$K:=\bar PC^{\top}S^{-1}$ and $\Afil:=A(I-KC)$, which is Schur. A smart
sensor collocated with \eqref{eq:plant} runs the steady-state Kalman
filter, producing the posterior estimate $\xs_{k}$ and the innovation
\begin{equation}
z_{k}=y_{k}-C\xs_{k|k-1},\qquad \Exp[z_{k}z_{k}^{\top}]=S,
\label{eq:innov}
\end{equation}
with prior error $e_{k}:=x_{k}-\xs_{k|k-1}$, so that
$e_{k+1}=\Afil e_{k}+w_{k}-AKv_{k}$ and $z_{k}=Ce_{k}+v_{k}$. The remote
estimator applies the received innovation $\zc_{k}$ directly,
\begin{equation}
\xa_{k}=\xa_{k|k-1}+K\zc_{k},\qquad \xa_{k+1|k}=A\xa_{k}.
\label{eq:remote}
\end{equation}

\begin{assumption}
\label{as:A1}
\emph{(i)} The network carries the innovation and the remote node
updates by \eqref{eq:remote}. \emph{(ii)} The sensor computes
\eqref{eq:innov} from the true measurements, and $\xa_{k}$ is fed back
neither to the plant nor to the sensor filter. \emph{(iii)} The
measurement is scalar, $m=1$.
\end{assumption}

Part~(ii) makes the monitored signal \emph{exogenous}: corruption
accumulates at the remote node while the corrupted signal is generated
from the nominal stream, so no loop closes around the plant.
Lemmas~\ref{lem:magnitude} and~\ref{lem:damage} hold verbatim for $m>1$;
part~(iii) is used from \eqref{eq:admissible} onward, where it makes the
score $|z_{k}|$ and the exposure functional $\rlag_{\ell}$ scalar.

The adversary is a man-in-the-middle on the link: it observes the
nominal stream for $N$ steps and may replace transmitted packets
thereafter, as shown in Fig.~\ref{fig:attack}. It knows none of $(A,C,Q,R)$, $S$, $K$, the remote
estimator's state, the detector or its threshold, and it does not know
the innovation law. With $\gamma_{k}\in\{0,1\}$ the firing indicator,
its resource limit is the budget
\begin{equation}
\limsup_{T\to\infty}\tfrac1T\textstyle\sum_{k=1}^{T}\gamma_{k}\le\Gbar .
\label{eq:budget}
\end{equation}

\begin{figure}[t]\centering
\includegraphics[width=\columnwidth]{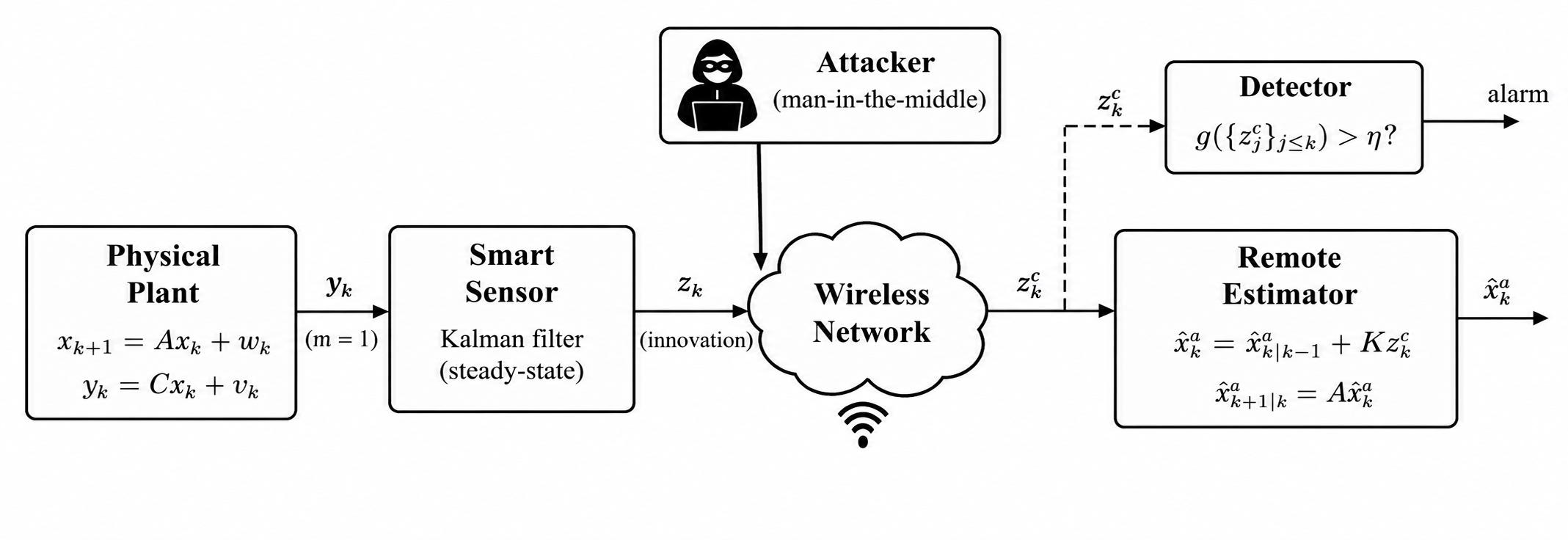}
\caption{System under Attack.}
\label{fig:attack}
\end{figure}

The link is monitored by a residual-based detector: a measurable
functional $g$ of $\{\zc_{j}\}_{j\le k}$  alarming when $g>\eta$, with $\eta$ calibrated on nominal data. The adversary does not know which detector is deployed, so we fix two classes: $\Gmag$, those measurable with respect to the magnitude process $\sigma(\|\zc_{j}\|:j\le k)$,
containing the memoryless $\chi^{2}$ test
$\zc_{k}{}^{\top}S^{-1}\zc_{k}>\eta$, its windowed and CUSUM variants
\cite{basseville1993detection} and the Serial Detector
\cite{bonczek2021detection}; and $\Gsgn$, those measurable with respect
to the full received sequence but not with respect to its magnitudes
alone --- lag-$\ell$ autocorrelation and Ljung--Box whiteness tests
\cite{mehra1971innovations,ljung1978measure}, and sign-balance tests on
$\sgn(\zc_{k})$.

\subsection{The Attack Action and Magnitude Stealth}
\label{sec:action}
Within the linear attack family $z^{a}_{k}=F_{k}z_{k}+b_{k}$,
$b_{k}\sim\mathcal{N}(0,\Sigma_{k})$ independent of $z_{k}$, under the
covariance-matching stealth constraint
$F_{k}SF_{k}^{\top}+\Cov(b_{k})=S$ standard in this literature
\cite{Guo2023,Wang2025,guo2018worst}, the maximizer of the remote error
covariance is $F^{\star}=-I$, $b\equiv0$ \cite[Thm.~3]{Guo2023}. At
firing instants the adversary therefore transmits the \emph{sign flip}
\begin{equation}
\zc_{k}=\vsg_{k}z_{k},\qquad \vsg_{k}:=1-2\gamma_{k}\in\{-1,+1\},
\label{eq:action}
\end{equation}
so that $\zc_{k}=-z_{k}$ when $\gamma_{k}=1$ and $\zc_{k}=z_{k}$
otherwise. Two features distinguish \eqref{eq:action} from the
measurement-space injections common in this literature. It is
\emph{model-free by construction}, the adversary negating a signal it
already reads, so no prediction, filter copy or covariance estimate is
needed; and it preserves magnitudes pathwise, which is the basis of its
stealth certificate.

\begin{lemma}[Magnitude stealth {\cite[Thm.~1]{din2026distributionfreebudgetedstealthyattack}}]
\label{lem:magnitude}
Under \eqref{eq:action}, $\|\zc_{k}\|=\|z_{k}\|$ for every $k$ on every
sample path, irrespective of the firing rule, the budget and the
distribution of $z_{k}$. Consequently every $g\in\Gmag$ satisfies
$g(\{\zc_{j}\})=g(\{z_{j}\})$ pathwise, and for every $\eta$ its
false-alarm rate under attack equals its nominal rate exactly.
\end{lemma}

Because $|\zc_{k}|=|z_{k}|$, any magnitude score is computable from the
adversary's own \emph{output} stream, so the attack cannot corrupt its
own trigger; and because the certificate is a magnitude identity, a rule
firing on signs falls outside the guaranteed class. We therefore admit only
\emph{magnitude-measurable} schedules,
\begin{equation}
\gamma_{k}=\phi\big(|z_{k-L}|,\dots,|z_{k}|\big)\in\{0,1\},
\label{eq:admissible}
\end{equation}
for a causal window of length $L\ge0$; $L=0$ is the memoryless case. The
same $L$ indexes the depth of the exposure constraint in
\eqref{eq:cproblem} and the maximum run length of \eqref{eq:hyst}: in
each case it is the number of past instants the schedule may consult.

The action \eqref{eq:action} acts on the signs of the innovation while
\eqref{eq:admissible} reads only its magnitudes, so the dependence
between signs and magnitudes governs both what the attack achieves and
what it reveals. The case in which that dependence is absent is the
reference point for everything below.

\begin{definition}
\label{def:css}
$\{z_{k}\}$ is \emph{conditionally sign symmetric} (CSS) if, given the
entire magnitude process $\{|z_{j}|\}_{j}$, the signs $\{\sgn(z_{k})\}$
are i.i.d.\ uniform on $\{\pm1\}$.
\end{definition}

Definition~\ref{def:css} holds automatically under Gaussian noise and
fails only off it, as Lemma~\ref{lem:boundary} makes precise.

\subsection{Damage and the Memoryless Benchmark}
\label{sec:benchmark}
Let $d_{k}:=\xs_{k}-\xa_{k}$ be the divergence between the clean local
estimate and the corrupted remote one. Subtracting \eqref{eq:remote}
from the sensor recursion and using $z_{k}-\zc_{k}=2\gamma_{k}z_{k}$,
\begin{equation}
d_{k}=Ad_{k-1}+K(z_{k}-\zc_{k})=Ad_{k-1}+2\gamma_{k}Kz_{k},
\label{eq:divergence}
\end{equation}
so each firing injects a kick $2Kz_{k}$ proportional to the innovation
at that instant, and past kicks decay through $A$. We measure the attack
by $\tr(D_{k})$ with $D_{k}:=\Exp[d_{k}d_{k}^{\top}]$: this is the
component of the remote error attributable to the attack, and it
vanishes in the attack's absence.

\begin{lemma}[Damage collapses to one scalar {\cite[Thm.~3]{din2026distributionfreebudgetedstealthyattack}}]
\label{lem:damage}
Let $\rho(A)<1$, let $\{z_{k}\}$ be stationary and conditionally sign
symmetric (Def.~\ref{def:css}), let $\gamma$ satisfy
\eqref{eq:admissible} with $\Prob(\gamma_{k}=1)=\Gbar$, and let
$D_{0}=0$. Then, with the \emph{energy capture}
$\Psi:=\Exp[\gamma_{k}z_{k}^{2}]/S$,
\begin{multline}
D_{\infty}=AD_{\infty}A^{\top}+4\Psi SKK^{\top}, \\
\tr(D_{\infty})=4S\Psi\,\tr\Lyap(A,KK^{\top}).
\label{eq:Dinf}
\end{multline}
\end{lemma}

Expanding \eqref{eq:divergence} produces cross terms in
$\Exp[d_{k-1}\gamma_{k}z_{k}^{\top}]$ besides the driving term; these
vanish under Def.~\ref{def:css}, since $d_{k-1}$ is past-measurable and
$\gamma_{k}$ reads magnitudes only. Off Def.~\ref{def:css} they need not
vanish and \eqref{eq:Dinf} ceases to be an identity, but $\Psi$ remains
the only schedule-dependent quantity in the driving term and is the
design objective throughout.

Since $4S$ and $\tr\Lyap(A,KK^{\top})$ are plant and filter properties
that no schedule can alter, the adversary's entire influence is the
scalar $\Psi$ --- the share of innovation energy it corrupts. It equals
$\Gbar$ for a schedule independent of $\{z_{j}\}$ and tends to one as
the firing set concentrates on the largest innovations.

\begin{lemma}[Memoryless optimality {\cite[Cor.~1]{din2026distributionfreebudgetedstealthyattack}}]
\label{lem:memoryless}
Let $F_{|z|}$ be continuous. For every schedule satisfying
\eqref{eq:admissible} with rate $\Gbar$,
\begin{equation}
\Psi=\Gbar+\Cov(\gamma_{k},z_{k}^{2})/S,
\label{eq:psicov}
\end{equation}
and $\Psi$ is maximized by the \emph{memoryless} upper level set
\begin{equation}
\gamma^{\star}_{k}=\mathbf1\{|z_{k}|>q\},\qquad
q:=F^{-1}_{|z|}(1-\Gbar),
\label{eq:optrule}
\end{equation}
any maximizer agreeing with $\gamma^{\star}$ almost everywhere.
\end{lemma}

Lemma~\ref{lem:memoryless} is the benchmark from which this paper
departs, and it shows why the departure is not obvious: absent any
constraint from $\Gsgn$, the optimum uses no memory, the objective being
pointwise in $\gamma$ and a pointwise objective being maximized by
sorting.

The threshold $q$ in \eqref{eq:optrule} depends on $F_{|z|}$, which the
adversary does not know. It is obtained instead from the nominal record,
without distributional knowledge. Let $s_{1},\dots,s_{N}$ be a nominal
calibration record of a magnitude-measurable score, and let
$\alpha\in[1/(N+1),1]$ be a \emph{conformal level}, that is, a nominal
exceedance probability. The split-conformal threshold at level $\alpha$
is the order statistic
\begin{equation}
\hat\varepsilon_{N}(\alpha):=s_{(j_{N}(\alpha))},\qquad
j_{N}(\alpha):=\lceil(N+1)(1-\alpha)\rceil .
\label{eq:conformal}
\end{equation}

\begin{lemma}[Distribution-free level {\cite[Thm.~2]{din2026distributionfreebudgetedstealthyattack}}]
\label{lem:budget}
If $s_{1},\dots,s_{N+1}$ are exchangeable then
\begin{equation}
\Prob\big(s_{N+1}>\hat\varepsilon_{N}(\alpha)\big)
=1-j_{N}(\alpha)/(N+1)\le\alpha .
\label{eq:conformalguarantee}
\end{equation}
If in addition $\{s_{k}\}$ is stationary and ergodic, the realized
exceedance rate converges almost surely to
$1-F_{s}(\hat\varepsilon_{N}(\alpha))$, and to $\alpha$ as
$N\to\infty$.
\end{lemma}

The bound is the uniformity of the rank of $s_{N+1}$ among $N+1$
exchangeable scores; the limits follow from Birkhoff's and the ergodic
Glivenko--Cantelli theorems. Two properties matter here: \eqref{eq:conformal} needs exchangeability rather than independence, which is essential because off Gaussianity the innovation is uncorrelated but dependent, and the index $j_{N}$ is taken over $N+1$, which removes the over-firing bias of the empirical quantile at short records.

Applied to $s_{k}=|z_{k}|$ at level $\alpha=\Gbar$, the rule
$\gamma_{k}=\mathbf1\{|z_{k}|>\hat\varepsilon_{N}(\Gbar)\}$ recovers
\eqref{eq:optrule} with $\hat\varepsilon_{N}$ in place of $q$: for the
memoryless rule a single order statistic delivers both the budget
guarantee and the damage optimum. Once the schedule has memory, the
level $\alpha$ and the firing rate $\Gbar$ separate
(Sec.~\ref{sec:calibration}).

\subsection{The Stealth Boundary}
\label{sec:signs}
The innovation of a correctly tuned Kalman filter is white,
$\Exp[z_{k}z_{k+\ell}^{\top}]=0$ for $\ell\ge1$, whatever the noise law.
Under Gaussian noise whiteness upgrades to independence, so
Def.~\ref{def:css} holds; off Gaussianity the innovation is uncorrelated
but \emph{not} independent and the link between signs and magnitudes
survives. Lemma~\ref{lem:boundary} turns that link into the exact
boundary of the certificate of Lemma~\ref{lem:magnitude}, and names the
obstruction.

\begin{lemma}[Stealth boundary {\cite[Prop.~2]{din2026distributionfreebudgetedstealthyattack}}]
\label{lem:boundary}
Let $\gamma$ satisfy \eqref{eq:admissible}. If $\{z_{k}\}$ is CSS then
$\{\zc_{k}\}\overset{d}{=}\{z_{k}\}$ as processes, and every detector in
$\Gmag$ \emph{or} $\Gsgn$ retains its nominal false-alarm rate.
Conversely, CSS fails whenever
\begin{equation}
\Exp[z_{k}^{3}z_{k+1}]
=C\Afil\cum(z_{k},z_{k},z_{k},e_{k})-CAK\,\kappa_{4}(v_{k})\neq0,
\label{eq:fourth}
\end{equation}
and both terms of \eqref{eq:fourth} are fourth cumulants, vanishing for
Gaussian $w,v$. Exposure is thus a purely non-Gaussian phenomenon.
\end{lemma}

Under CSS the sign flip is invisible to \emph{every} residual-based
detector and there is nothing to design around. Off Gaussianity
\eqref{eq:fourth} is generically non-zero, the dependence it measures
lives in the innovation's tail, and \eqref{eq:optrule} fires precisely
there: damage and detectability are two readings of one tail quantity.

\subsection{The Exposure of a Schedule}
\label{sec:exposure}
The detectors in $\Gsgn$ are correlation tests, so to constrain the
exposure against these the natural object is the autocovariance of the
transmitted sequence.

\begin{definition}
\label{def:exposure}
For $\ell\ge1$ the \emph{lag-$\ell$ exposure} of a schedule is
$\rlag_{\ell}(\gamma):=\Exp[\zc_{k}\zc_{k+\ell}]/S$.
\end{definition}

Three properties recommend $\rlag_{\ell}$. It is the population quantity
the lag-$\ell$ correlation detectors estimate, and that Ljung--Box
\cite{ljung1978measure} aggregates across lags. It is \emph{estimable
without a model}: since $\Exp[z_{k}z_{k+\ell}]=0$ nominally and
$|\zc|=|z|$, the adversary evaluates $\hat\rlag_{\ell}$ on the very
stream it transmits and the distribution-free discipline of Lemma~\ref{lem:budget} survives. And it degenerates correctly
(Prop.~\ref{prop:exposure}).

Expanding \eqref{eq:action} with
$\vsg_{k}\vsg_{k+\ell}=1-2\gamma_{k}-2\gamma_{k+\ell}
+4\gamma_{k}\gamma_{k+\ell}$ and using whiteness of the nominal
innovation,
\begin{multline}
\rlag_{\ell}(\gamma)\,S=
-2\Exp[\gamma_{k}z_{k}z_{k+\ell}]-2\Exp[\gamma_{k+\ell}z_{k}z_{k+\ell}]\\
+4\Exp[\gamma_{k}\gamma_{k+\ell}z_{k}z_{k+\ell}].
\label{eq:exposure}
\end{multline}
Writing $A_{k}:=\{\gamma_{k}=1\}$, the case $\ell=1$ collapses to
\begin{equation}
\rlag_{1}(\gamma)\,S=-2\,(T_{10}+T_{01}),
\label{eq:induced}
\end{equation}
with $T_{10}:=\Exp[\mathbf1_{A_{k}}\mathbf1_{A^{c}_{k+1}}z_{k}z_{k+1}]$
and $T_{01}:=\Exp[\mathbf1_{A^{c}_{k}}\mathbf1_{A_{k+1}}z_{k}z_{k+1}]$,
the concordant terms cancelling: only neighboring pairs in which one
instant fires and the other does not contribute to the lag-one exposure.

\begin{proposition}[Degeneracy under sign symmetry]
\label{prop:exposure}
If $\{z_{k}\}$ is CSS then $\rlag_{\ell}(\gamma)=0$ for every $\ell\ge1$
and every $\gamma$ satisfying \eqref{eq:admissible}. Conversely, each
term of \eqref{eq:exposure} is a magnitude-truncated cross moment,
non-zero only through the conditional sign correlation
$\Exp[\sgn(z_{k})\sgn(z_{k+\ell})\mid\{|z_{j}|\}]$, whose leading scalar
witness at $\ell=1$ is \eqref{eq:fourth}.
\end{proposition}

\begin{proof}
Conditioning on the magnitude process, each expectation in
\eqref{eq:exposure} carries the factor
$\Exp[\sgn(z_{k})\sgn(z_{k+\ell})\mid\{|z_{j}|\}]$, which is zero under
Def.~\ref{def:css} since the conditional signs are i.i.d.\ uniform and
$\gamma$ depends on magnitudes alone; the converse is
Lemma~\ref{lem:boundary}.
\end{proof}

Proposition~\ref{prop:exposure} delimits the scope of this paper: under
Gaussian noise the constraint introduced next is inactive for
\emph{every} admissible schedule and $\gamma^{\star}$ remains optimal.
The constrained problem is a strictly non-Gaussian object.

\subsection{Problem Statement}
\label{sec:statement}
The adversary's two design choices are of different kinds. The
\emph{signal} is determined: the sign flip maximizes damage within its
family \cite[Thm.~3]{Guo2023} and, by Lemma~\ref{lem:magnitude}, is
exactly stealthy against $\Gmag$ whatever the schedule. Every remaining
degree of freedom lies in the \emph{schedule} $\{\gamma_{k}\}$, which
must serve two distinct requirements: setting the damage through $\Psi$
(Lemma~\ref{lem:damage}) and the exposure through $\rlag_{\ell}$
(Def.~\ref{def:exposure}). Fixing a tolerance
$\delta=(\delta_{1},\dots,\delta_{L})$, $\delta_{\ell}\ge0$, on the
exposure the adversary is willing to present, we solve
\begin{equation}
\begin{aligned}
\max_{\gamma}\ \ & \Psi(\gamma)=\Exp[\gamma_{k}z_{k}^{2}]/S\\
\text{s.t.}\ \ & \Exp[\gamma_{k}]\le\Gbar,\quad
  \gamma_{k}=\phi(|z_{k-L}|,\dots,|z_{k}|)\in\{0,1\},\\
& |\rlag_{\ell}(\gamma)|\le\delta_{\ell},\quad \ell=1,\dots,L,
\end{aligned}
\label{eq:cproblem}
\end{equation}
the three constraints being the budget, admissibility and the exposure
tolerance respectively. By \eqref{eq:Dinf}, maximizing $\Psi$ maximizes
$\tr(D_{\infty})$, so \eqref{eq:cproblem} asks for the most damaging
schedule whose transmitted stream is, to within $\delta$,
indistinguishable from nominal to a correlation monitor. All three
constraints are checkable by the adversary from magnitudes alone.

Three features of \eqref{eq:cproblem} determine the analysis that
follows. First, the objective is \emph{pointwise} in $\gamma$ while the
exposure constraint is \emph{bilinear across two instants}, so the
problem does not separate instant by instant and cannot be solved by
sorting, and Lemma~\ref{lem:memoryless} cannot be expected to survive.
Second, the window $L$ is imposed by the constraint rather than chosen
freely, and Sec.~\ref{sec:price} prices it. Third, $\delta$ indexes a
frontier: at $\delta=\infty$ the problem returns $\gamma^{\star}$, while
at $\delta=0$ any $\gamma$ independent of $\{z_{j}\}$ remains feasible
with $\Psi=\Gbar$ by \eqref{eq:psicov}, so the blind point $(\Gbar,0)$
anchors the frontier from below at every $\delta$. 
How much damage
survives a small $\delta$ off Gaussianity is the question
Sec.~\ref{sec:structure} answers.

\section{Structure of the Constrained Optimum}
\label{sec:structure}

\subsection{Damage and Exposure Are Distinct Functionals}
\label{sec:decouple}

\begin{proposition}[Decomposition]
\label{prop:decouple}
Let $\gamma$ satisfy \eqref{eq:admissible} with rate $\Gbar$. Then
\begin{enumerate}
\item[(i)] by \eqref{eq:psicov}, $\Psi$ depends on the firing set only
  through its \emph{location} in magnitude space, and is maximized by
  placing it on the upper $\Gbar$-tail;
\item[(ii)] by \eqref{eq:induced}, $\rlag_{1}$ depends on the firing set
  only through pairs $(k,k+1)$ that \emph{straddle} its boundary:
  \begin{equation}
  \rlag_{1}(\gamma)=-\tfrac{2}{S}\,
  \Prob(\gamma_{k}\neq\gamma_{k+1})\;
  \Exp[z_{k}z_{k+1}\mid\gamma_{k}\neq\gamma_{k+1}],
  \label{eq:factorized}
  \end{equation}
  a product of a \emph{boundary density} and a conditional tail cross
  moment;
\item[(iii)] writing $\rho_{1}$ for the lag-one autocorrelation of the
  firing indicator $\{\gamma_{k}\}$, the boundary density is
  \begin{equation}
  \Prob(\gamma_{k}\neq\gamma_{k+1})=2\Gbar(1-\Gbar)(1-\rho_{1});
  \label{eq:boundarydensity}
  \end{equation}
  for a memoryless rule $\gamma_{k}=\mathbf1\{|z_{k}|>t\}$ with
  threshold $t$, the single parameter $t$ fixes \emph{both} the location
  and, through $\rho_{1}$, the boundary density.
\end{enumerate}
\end{proposition}

\begin{proof}
(i) is \eqref{eq:psicov}. (ii) Both surviving terms of
\eqref{eq:induced} carry the indicator
$\mathbf1\{\gamma_{k}\neq\gamma_{k+1}\}$, so
$T_{10}+T_{01}=\Exp[\mathbf1\{\gamma_{k}\neq\gamma_{k+1}\}z_{k}z_{k+1}]$,
and \eqref{eq:factorized} is the tower property. (iii) For a stationary
binary sequence,
$\Prob(\gamma_{k}=\gamma_{k+1}=1)=\Gbar^{2}+\rho_{1}\Gbar(1-\Gbar)$, and
$\Prob(\gamma_{k}\neq\gamma_{k+1})=2[\Gbar-\Prob(\gamma_{k}=\gamma_{k+1}=1)]$
gives \eqref{eq:boundarydensity}; a memoryless rule fixes the level set
$\{|z|>t\}$ and, $t$ being its only parameter, fixes the joint law of
$(\gamma_{k},\gamma_{k+1})$ with it.
\end{proof}

Proposition~\ref{prop:decouple} is a statement about degrees of freedom.
Damage is a property of \emph{where} the firing set sits, exposure of
\emph{how it is arranged in time}; these are independent attributes of a
subset of the time axis, yet a memoryless rule is indexed by the single
scalar $t$ that fixes both. Raising $t$ concentrates the firing set on
the tail, which \eqref{eq:Dinf} rewards, and makes every firing an
isolated spike flanked by two boundaries, which \eqref{eq:induced}
penalizes; off Gaussianity the two effects reinforce. Breaking the
coupling needs a second degree of freedom that holds the location fixed
while reorganizing runs, and such a rule must consult past decisions.

Equation \eqref{eq:boundarydensity} also quantifies what memory can buy.
A rule that fires in runs of length $L$ has
$\rho_{1}=\big[(L-1)/L-\Gbar\big]/(1-\Gbar)$, which approaches $1-1/L$
at small budgets; its boundary density is then $2\Gbar(1-\Gbar)/L$, so
run structure suppresses exposure as $1/L$. The same $\rho_{1}$
reappears in Prop.~\ref{prop:neff} as the contraction of the effective
calibration sample, so the benefit and the price of memory are governed
by one quantity.

\subsection{Randomization and Windowed Scores Do Not Decouple Them}
\label{sec:dither}

Proposition~\ref{prop:decouple} predicts, before any experiment, the
behavior of the two families a designer would reach for first.

\emph{Randomized scores.} With $\xi_{k}$ i.i.d.\ zero-mean and
unit-variance, independent of $\{z_{j}\}$, and a randomization scale
$\sigma_{d}\ge0$, set $s_{k}=|z_{k}|+\sigma_{d}\xi_{k}$. This score
remains magnitude measurable, leaves
Lemmas~\ref{lem:magnitude}--\ref{lem:budget} intact, and interpolates
between $\gamma^{\star}$ at $\sigma_{d}=0$ and blind firing as
$\sigma_{d}\to\infty$.

\begin{remark}[Randomization is not a stealth parameter]
\label{rem:dither}
Along this family $\Psi$ is non-increasing in $\sigma_{d}$ with
$\Psi\to\Gbar$, whereas $\rlag_{1}$ is in general \emph{not} monotone
and changes sign at a value of $\sigma_{d}$ that depends on the
innovation law. The mechanism is Prop.~\ref{prop:decouple}: randomizing
the score perturbs the \emph{location} of the firing set, which is the
term carrying $\Psi$, while acting on exposure only through the
conditional cross moment of Prop.~\ref{prop:decouple}(ii) --- a
difference of truncated moments whose weight shifts, and whose sign
inverts, as the firing set slides off the tail. The $\sigma_{d}$ at
which exposure is small therefore depends on the innovation law and on
the budget, so randomization buys correlation stealth only at the cost
of the distributional knowledge this framework is built to avoid.
\end{remark}

\emph{Off-the-shelf windowed scores.} A score concentrated on
near-extreme instants --- a one-step magnitude predictor, or a
magnitude-difference score --- preserves the location of the firing set,
hence $\Psi$, but preserves its isolated-spike run structure as well,
hence the exposure. A windowed aggregate --- a moving energy, or a fixed
hold after a crossing --- suppresses boundary density and with it the
exposure, but does so only by dispersing the firing set off the tail, so
that $\Psi$ collapses toward $\Gbar$. Neither family decouples the two
functionals; they occupy opposite ends of a single trade-off, and
Sec.~\ref{sec:results} measures both.

\subsection{The Corrected Score}
\label{sec:corrected}

We dualize the exposure constraint and evaluate its effect on a single
decision. Introduce a multiplier
$\lambda\in\mathbb R$ for the budget constraint and
$\mu=(\mu_{1},\dots,\mu_{L})$ for the exposure constraints, where
$\mu_{\ell}:=\mu^{+}_{\ell}-\mu^{-}_{\ell}\in\mathbb R$ collects the
non-negative multipliers of the two sides $\rlag_{\ell}\le\delta_{\ell}$
and $-\rlag_{\ell}\le\delta_{\ell}$, of which at most one is active. The
Lagrangian is
\begin{equation}
\mathcal J(\gamma;\lambda,\mu)=\Exp[\gamma_{k}z_{k}^{2}]
-\lambda\Exp[\gamma_{k}]
-\textstyle\sum_{\ell=1}^{L}\mu_{\ell}\,S\,\rlag_{\ell}(\gamma).
\label{eq:lagrangian}
\end{equation}

Optimizing \eqref{eq:lagrangian} instant by instant requires only the
effect of a single decision on the exposure. Holding all other decisions
fixed and switching $\gamma_{k}$ from $0$ to $1$ changes $\zc_{k}$ from
$z_{k}$ to $-z_{k}$ and leaves every other transmitted sample
unaltered, so the only affected terms of
$S\rlag_{\ell}=\Exp[\zc_{k}\zc_{k+\ell}]$ are the two products in which
$\zc_{k}$ appears, namely those indexed by the pairs $(k-\ell,k)$ and
$(k,k+\ell)$. Hence
\begin{equation}
S\,\frac{\partial\rlag_{\ell}}{\partial\gamma_{k}}
=-2\,z_{k}\big(\zc_{k-\ell}+\zc_{k+\ell}\big).
\label{eq:flipdiff}
\end{equation}

\begin{theorem}[Structure of the constrained optimum]
\label{thm:structure}
Let $\Mset_{k}:=\sigma(|z_{j}|:j\le k)$ be the magnitude filtration and
let $\gamma$ be admissible for \eqref{eq:cproblem}. At any solution
there exist multipliers $(\lambda,\mu)$ such that, almost everywhere,
\begin{equation}
\gamma_{k}=\mathbf1\{\,z_{k}^{2}>\lambda+c_{k}\,\},\qquad
c_{k}=\sum_{\ell=1}^{L}\mu_{\ell}\,
\Exp\!\left[S\,\frac{\partial\rlag_{\ell}}{\partial\gamma_{k}}
\,\Big|\,\Mset_{k}\right],
\label{eq:corrected}
\end{equation}
that is, the optimum is a threshold on the \emph{corrected score}
$z_{k}^{2}-c_{k}$. Moreover $c_{k}$ is $\Mset_{k}$-measurable, so
\eqref{eq:corrected} satisfies \eqref{eq:admissible} and
Lemma~\ref{lem:magnitude} applies to it; and $c_{k}\equiv0$ whenever
$\mu=0$ or $\{z_{k}\}$ is conditionally sign symmetric, in either case
reducing \eqref{eq:corrected} to the memoryless rule $\gamma^{\star}$ of
\eqref{eq:optrule}.
\end{theorem}

\begin{proof}
Write $p_{k}:=\Exp[\gamma_{k}\mid\Mset_{k}]\in[0,1]$, over which
\eqref{eq:admissible} optimizes. With the remaining decisions fixed,
both the objective and \eqref{eq:exposure} are affine in $p_{k}$, so the
stationarity condition of \eqref{eq:lagrangian} is a pointwise
comparison whose optimizer is the upper level set of the coefficient of
$p_{k}$, which by \eqref{eq:flipdiff} is $z_{k}^{2}-\lambda-c_{k}$;
$\Mset_{k}$-measurability holds because $c_{k}$ is a conditional
expectation given $\Mset_{k}$. For the degeneracy, $\mu=0$ gives
$c_{k}\equiv0$. Otherwise, since $\vsg_{k-\ell}$ and $|z_{j}|$, $j\le
k$, are $\Mset_{k}$-measurable,
$\Exp[z_{k}\zc_{k-\ell}\mid\Mset_{k}]
=\vsg_{k-\ell}|z_{k}||z_{k-\ell}|\chi_{\ell}(k)$ with
$\chi_{\ell}(k):=\Exp[\sgn(z_{k})\sgn(z_{k-\ell})\mid\Mset_{k}]$, and
likewise for the forward term. Under Def.~\ref{def:css} the conditional
signs are i.i.d.\ uniform, so $\chi_{\ell}\equiv0$ and every term of
$c_{k}$ vanishes; \eqref{eq:corrected} is then the upper level set of
$|z_{k}|$ at the quantile fixed by $\lambda$, namely $\gamma^{\star}$.
\end{proof}

\begin{remark}[Scope of Theorem~\ref{thm:structure}]
\label{rem:fixedpoint}
Equation \eqref{eq:exposure} is bilinear in $\gamma$, so $c_{k}$ depends
on decisions that themselves depend on $c$: \eqref{eq:corrected} is a
necessary condition characterizing the structure of an optimum, not a
construction of one. The forward term $\zc_{k+\ell}$ in
\eqref{eq:flipdiff} is moreover unavailable to a causal scheduler and is
dropped in Sec.~\ref{sec:algorithm}, which gives the rule that is
deployed and evaluated.
\end{remark}

\begin{corollary}[Memory is necessary]
\label{cor:memory}
The memoryless rule $\gamma^{\star}$ solves \eqref{eq:cproblem} if and
only if it is feasible for it. Under conditional sign symmetry it is
feasible at every $\delta$, and off it, for
$\delta_{\ell}<|\rlag_{\ell}(\gamma^{\star})|$, it is not: any solution
then differs from $\gamma^{\star}$ on a set of positive measure and, by
Thm.~\ref{thm:structure}, thresholds a score whose correction $c_{k}$ is
not identically zero, hence depends on past decisions.
\end{corollary}

\begin{proof}
By Lemma~\ref{lem:memoryless}, $\gamma^{\star}$ maximizes $\Psi$ over
the admissible class at rate $\Gbar$, which contains the feasible set of
\eqref{eq:cproblem}; a feasible maximizer over a superset is a maximizer
over the subset, and an infeasible rule is not a solution. Feasibility
under Def.~\ref{def:css} is Prop.~\ref{prop:exposure}. The last claim is
Thm.~\ref{thm:structure} with $\mu\neq0$ and $\chi_{\ell}\not\equiv0$.
\end{proof}

Corollary~\ref{cor:memory} is the central structural claim of this
paper: a pointwise objective is maximized without memory
(Lemma~\ref{lem:memoryless}), while the same objective under the
exposure constraint is not, and the crossover between the two is exactly
conditional sign symmetry, hence exactly Gaussianity.

\subsection{A Causal Two-Threshold Realization}
\label{sec:algorithm}

Theorem~\ref{thm:structure} states that the optimum thresholds a score
carrying a correction built from past magnitudes and past decisions,
without specifying its form. Proposition~\ref{prop:decouple} determines
that form. By \eqref{eq:factorized} the lag-one exposure factors as
$|\rlag_{1}|=2S^{-1}\beta\,|m_{B}|$, with
$\beta:=\Prob(\gamma_{k}\neq\gamma_{k+1})$ the boundary density and
$m_{B}:=\Exp[z_{k}z_{k+1}\mid\gamma_{k}\neq\gamma_{k+1}]$ the
conditional cross moment on the boundary. Of the two factors only $\beta$ is under the schedule's direct control at a fixed budget, and by
\eqref{eq:boundarydensity} it falls as the firing indicator becomes more
persistent. At a fixed rate $\Gbar$ and a fixed firing-set location,
therefore, the exposure constraint is relaxed in exactly one way: by
\emph{lengthening runs}, which reduces the number of boundaries the same
number of firings must create.

This identifies the minimal admissible correction. Leaving the threshold
unchanged creates an isolated firing at every crossing and attains the
largest $\beta$ available at that budget; lowering it after a firing
extends crossings into runs and reduces $\beta$ without relocating the
firing set. Such a rule consults only past magnitudes and past
decisions, so it satisfies \eqref{eq:admissible}, and it is a
hysteresis: a high threshold fixing where the firing set sits, and a
lower one fixing how it is arranged in time --- the second degree of freedom the memoryless class lacks.

We deploy the minimal causal rule with this structure. Fix a window $L$
and a \emph{continuation gate} $\fcont\in[0,1]$, let
$\varepsilon_{\mathrm{hi}}$ be set by calibration
(Sec.~\ref{sec:calibration}) and put
$\varepsilon_{\mathrm{lo}}=\fcont\,\varepsilon_{\mathrm{hi}}$. With
$n_{k-1}$ the length of the run in progress at time $k-1$,
\begin{equation}
\gamma_{k}=
\begin{cases}
1, & |z_{k}|>\varepsilon_{\mathrm{hi}}\quad\text{(start)},\\
1, & 0<n_{k-1}<L\ \text{and}\ |z_{k}|>\varepsilon_{\mathrm{lo}}
      \quad\text{(continue)},\\
0, & \text{otherwise}.
\end{cases}
\label{eq:hyst}
\end{equation}
Rule \eqref{eq:hyst} realizes \eqref{eq:corrected} with a negative
correction --- the threshold falls from $\varepsilon_{\mathrm{hi}}$ to
$\varepsilon_{\mathrm{lo}}$ --- exactly when the past $L$ decisions
place the instant inside a run. It reads only past magnitudes and past
decisions, so \eqref{eq:admissible} and Lemma~\ref{lem:magnitude} still
apply and the attack remains exactly stealthy against $\Gmag$. At
$\fcont=1$ there is no continuation and \eqref{eq:hyst} is the
memoryless rule; as $\fcont\to0$ every start runs to length $L$, giving
a fixed hold; intermediate $\fcont$ traverses the frontier.
Algorithm~\ref{alg:corr} keeps one counter and makes at most two
comparisons per step with no plant matrix.

\begin{algorithm}[t]
\caption{Correlation-stealthy budgeted scheduling}
\label{alg:corr}
\begin{algorithmic}[1]
\Require budget $\Gbar$, window $L$, gate $\fcont$, nominal record
         $z_{1},\dots,z_{N}$
\State choose the level $\alpha$ such that \eqref{eq:hyst} fires at rate
       $\Gbar$ on the record; set
       $\varepsilon_{\mathrm{hi}}\gets\hat\varepsilon_{N}(\alpha)$ by
       \eqref{eq:conformal} applied to $\{|z_{k}|\}$, and
       $\varepsilon_{\mathrm{lo}}\gets\fcont\,\varepsilon_{\mathrm{hi}}$
\State $n\gets0$
\For{$k=1,2,\dots$}
  \State $n\gets 1$ if $|z_{k}|>\varepsilon_{\mathrm{hi}}$; else
         $n\gets n+1$ if $0<n<L$ and
         $|z_{k}|>\varepsilon_{\mathrm{lo}}$; else $n\gets0$
  \State $\gamma_{k}\gets\mathbf1\{n>0\}$;\quad transmit
         $\zc_{k}\gets(1-2\gamma_{k})z_{k}$
\EndFor
\end{algorithmic}
\end{algorithm}

\section{Calibration, Budget Guarantee and the Price of Memory}
\label{sec:calibration}

\subsection{Where the Budget Guarantee Holds}
\label{sec:budgetguarantee}

The memoryless rule and the two-threshold rule stand in different
relations to Lemma~\ref{lem:budget}. For the memoryless rule the firing
event \emph{is} the event $\{|z_{k}|>\hat\varepsilon_{N}\}$, a threshold
on an exchangeable scalar score, and \eqref{eq:conformalguarantee}
applies to it directly. For \eqref{eq:hyst} the firing event is a union
of a threshold event and a continuation event, and the latter depends on
the run state $n_{k-1}$. The following proposition separates the two.

\begin{proposition}[Start-rate guarantee]
\label{prop:startrate}
Let $\varepsilon_{\mathrm{hi}}=\hat\varepsilon_{N}(\alpha)$ be the
split-conformal threshold \eqref{eq:conformal} applied to the nominal
magnitude record $\{|z_{k}|\}_{k\le N}$, and let
$\mathrm{st}_{k}:=\mathbf1\{|z_{k}|>\varepsilon_{\mathrm{hi}}\}$ denote
a run start. If $\{|z_{k}|\}$ is exchangeable then
\begin{equation}
\Exp[\mathrm{st}_{k}]\le\alpha
\label{eq:startguarantee}
\end{equation}
at every horizon, distribution-free and finite-sample. The realized
firing rate $\hat\Gamma:=\Exp[\gamma_{k}]$ of \eqref{eq:hyst} satisfies
\begin{equation}
\hat\Gamma\le L\,\alpha \quad\text{pathwise},\qquad
\hat\Gamma\approx\alpha\,\Exp[B],
\label{eq:ratebounds}
\end{equation}
where $\Exp[B]$ is the mean run length, the second relation holding in
the renewal limit.
\end{proposition}

\begin{proof}
The start event is a threshold on the scalar score $|z_{k}|$ at the
order statistic $\hat\varepsilon_{N}(\alpha)$, so
\eqref{eq:startguarantee} is \eqref{eq:conformalguarantee} verbatim.
For the first bound in \eqref{eq:ratebounds}, each start licenses at
most $L-1$ continuations by the second branch of \eqref{eq:hyst}, so on
every sample path the number of firings in any interval is at most $L$
times the number of starts. The second relation is the renewal-reward
identity for an alternating sequence of runs and gaps.
\end{proof}

Proposition~\ref{prop:startrate} locates the guarantee precisely. It
holds for the run-start rate, exactly and without distributional
assumptions, at whatever level $\alpha$ the order statistic is taken. It
does not hold for the realized firing rate, because $\gamma_{k}$ is not
a threshold on an exchangeable scalar; what the realized rate inherits
is the pathwise bound $\hat\Gamma\le L\alpha$ and the renewal relation
$\hat\Gamma\approx\alpha\Exp[B]$, in which $\Exp[B]$ depends on the
innovation law and is therefore not available in closed form to a
distribution-free adversary.

This determines the choice of $\alpha$. Setting $\alpha=\Gbar$ and
reading \eqref{eq:ratebounds} shows that the rule would then fire at
approximately $\Gbar\,\Exp[B]$, overshooting the budget by the mean run
length. Line~1 of Algorithm~\ref{alg:corr} therefore calibrates the
\emph{level}: $\alpha$ is chosen so that \eqref{eq:hyst} meets the rate
$\Gbar$ on the nominal record, with $\varepsilon_{\mathrm{hi}}$ kept an
order statistic of that record at the selected $\alpha$. The start
guarantee \eqref{eq:startguarantee} is retained exactly; the realized
rate becomes a plug-in quantity whose accuracy is governed by
Prop.~\ref{prop:neff}.

\subsection{The Price of Memory}
\label{sec:price}

\begin{proposition}[Effective calibration sample]
\label{prop:neff}
The guarantee of Lemma~\ref{lem:budget} requires exchangeability and not
independence, and is in that sense unaffected by memory: by
Prop.~\ref{prop:startrate} it continues to hold for the run-start rate
at every window length $L$. The \emph{variance} of the realized rate is
affected. Calibration contributes a standard deviation
$\sqrt{\Gbar(1-\Gbar)/\Neff}$ with
$\Neff=N/(1+2\sum_{\ell\ge1}\rho_{\ell})$, where $\rho_{\ell}$ is the
lag-$\ell$ autocorrelation of the firing indicator. A rule firing in runs of length $L$ has
$\rho_{\ell}\approx1-\ell/L$ for $\ell<L$, whence $\Neff\approx N/L$.
\end{proposition}

\begin{proof}
The first statement is Prop.~\ref{prop:startrate}. For the variance, the
realized rate is $1-F_{s}(\hat\varepsilon_{N})$ with
$\hat\varepsilon_{N}$ the $j_{N}$-th order statistic of the calibration
record, so for exchangeable scores
$F_{s}(\hat\varepsilon_{N})\sim\mathrm{\beta}(j_{N},N+1-j_{N})$ and the
realized rate has variance $\Gbar(1-\Gbar)/(N+2)$. Serial dependence
replaces $N$ by $\Neff$ via the long-run-variance expansion, and the
triangular autocorrelation $\rho_{\ell}=1-\ell/L$ gives
$1+2\sum_{\ell<L}(1-\ell/L)=L$.
\end{proof}

The mechanism is that overlapping windows make consecutive firing
decisions redundant: a record of $N$ nominal samples carries only
$\Neff$ independent-equivalent decisions, and it is $\Neff$, not $N$,
that sets the dispersion of the realized rate about the budget.

\subsection{Choosing the Window}
\label{sec:window}
By Prop.~\ref{prop:neff} the threshold $\varepsilon_{\mathrm{hi}}$ is
located not by the size of the calibration record but by the number of
exceedances within it, $\alpha\Neff\approx N\Gbar/L$, so a deep window
exhausts a short record at small budgets. Two ceilings on $L$ follow.
Run structure suppresses exposure as $1/L$ only while $\rlag_{\ell}$ is
materially non-zero at lag $\ell\le L$, so beyond the correlation length
a deeper window costs $\Psi$ and buys nothing; and $L$ must leave enough
effective exceedances to locate the $(1-\alpha)$-quantile, giving
$L\lesssim N\Gbar$. The first ceiling is the tighter of the two.

\section{Results}
\label{sec:results}

\emph{Setup:} The plant is
$A=\left[\begin{smallmatrix}0.95&0.02\\0&0.90\end{smallmatrix}\right]$,
$C=[1\ \ 0]$, $Q=0.01I_{2}$, $R=0.05$, giving $S=0.0752$,
$K=[0.335\ \ 0.029]^{\top}$ and $\tr\Lyap(A,KK^{\top})=1.182$. Both noises
are Gaussian under \textbf{N1}; under \textbf{N2} the sensor noise is the
variance-matched mixture
$0.95\,\mathcal N(0,0.0227)+0.05\,\mathcal N(0,0.568)$, with
$\kexc(z)=7.49$ against $0.00$ for N1. The two regimes present the same
$S$, so distributional shape is the only manipulated variable.
Calibration uses $N=5000$ nominal samples, deployment $T=50{,}000$ steps,
and $50$ Monte-Carlo realizations.

Every schedule is reported by two numbers: the energy capture $\Psi$ of
\eqref{eq:Dinf}, and $\pi$, the power of a lag-one innovation-whiteness
monitor of window $W=5000$ against a $1\%$ null calibrated empirically on
nominal data. The asymptotic null is unusable here, because by
Prop.~\ref{prop:exposure} the nominal innovation is white but not
independent and the heavy tail of N2 inflates the variance of the
statistic. All schedules are magnitude measurable, so by
Lemma~\ref{lem:magnitude} every detector in $\Gmag$ retains its nominal
false-alarm rate exactly and is not reported. Two references bracket the
comparison: the \emph{blind} rule, which fires independently of the
innovation at rate $\Gbar$ and captures $\Psi=\Gbar$, and the
\emph{memoryless} rule \eqref{eq:optrule}, which maximizes $\Psi$ under no
exposure constraint.

\subsection{The Corrected Rule at Matched Budget}

Table~\ref{tab:main} is the principal result. The memoryless rule is
detected at every budget, with $\pi$ from $0.128$ to $1.000$; the blind
rule is undetectable but captures only $\Gbar$. The corrected rule
\eqref{eq:hyst}, reported at its least-exposure $(L,\fcont)$ for each
budget, retains $80$--$96\%$ of the memoryless damage while cutting the
monitor power by factors of $2.3$ to $71$. The reduction is largest where
the memoryless rule is most exposed, at $\Gbar=0.02$, so the advantage of
correction is greatest at the small budgets to which a resource-limited
adversary is confined.

\begin{table}[h]\centering
\caption{The blind and memoryless rules bracket the achievable set; the
corrected rule is \eqref{eq:hyst} at its least-exposure $(L,\fcont)$.
$\Psi^{\star}$ and $\pi^{\star}$ are the memoryless values (N2, $W=5000$).}
\label{tab:main}
\setlength{\tabcolsep}{1.8pt}\footnotesize
\begin{tabular}{@{}l rr rr c rr rr@{}}
\toprule
& \multicolumn{2}{c}{blind}
& \multicolumn{2}{c}{memoryless \cite{din2026distributionfreebudgetedstealthyattack}}
& \multicolumn{5}{c}{corrected \eqref{eq:hyst}} \\
\cmidrule(lr){2-3}\cmidrule(lr){4-5}\cmidrule(lr){6-10}
$\Gbar$ & $\Psi$ & $\pi$ & $\Psi^{\star}$ & $\pi^{\star}$
        & $(L,\fcont)$ & $\Psi$ & $\pi$
        & $\tfrac{\Psi}{\Psi^{\star}}$ & $\tfrac{\pi^{\star}}{\pi}$ \\
\midrule
$0.02$ & $0.020$ & $0.010$ & $0.343$ & $1.000$ & $(4,0.30)$  & $0.294$ & $0.014$ & $0.86$ & $71.3$ \\
$0.05$ & $0.051$ & $0.010$ & $0.480$ & $0.369$ & $(6,0.50)$  & $0.455$ & $0.020$ & $0.95$ & $18.4$ \\
$0.10$ & $0.103$ & $0.016$ & $0.607$ & $0.255$ & $(10,0.15)$ & $0.485$ & $0.038$ & $0.80$ & $6.7$ \\
$0.20$ & $0.202$ & $0.024$ & $0.761$ & $0.858$ & $(10,0.15)$ & $0.626$ & $0.092$ & $0.82$ & $9.3$ \\
$0.30$ & $0.301$ & $0.040$ & $0.853$ & $0.745$ & $(10,0.15)$ & $0.724$ & $0.118$ & $0.85$ & $6.3$ \\
$0.50$ & $0.502$ & $0.028$ & $0.952$ & $0.128$ & $(4,0.30)$  & $0.915$ & $0.056$ & $0.96$ & $2.3$ \\
\bottomrule
\end{tabular}
\end{table}

\emph{Gaussian Recovery}: Under N1 the hypothesis of Prop.~\ref{prop:exposure} holds, the exposure
constraint is inactive for every admissible schedule, and memory should be unnecessary. It is: against the memoryless rule the whiteness
monitor stays at its nominal level, $\pi\in[0.000,0.024]$ at every
budget and window.

\subsection{Why Windowed Scores and Randomization Fail}

Table~\ref{tab:camps} compares the corrected rule against the alternatives
of Sec.~\ref{sec:dither} at exactly matched firing count. The scores fall
into two camps, as predicted. A one-step magnitude predictor and a
magnitude-difference score keep the firing set on the tail and so keep the
damage, but they fire in isolated instants and keep the exposure with it. A
windowed energy and a fixed hold remove the exposure, but only by firing in
long runs that disperse the firing set off the tail, and the damage
collapses to the blind value. The corrected rule is the only entry that
moves the two functionals separately.

\begin{table}[t]
\centering
\caption{Alternative scores at matched firing count
(N2, $\Gbar=0.02$, $W=5000$).}
\label{tab:camps}
\setlength{\tabcolsep}{3pt}\footnotesize
\begin{tabular}{@{}lrr@{}}
\toprule
Score & $\Psi$ & $\pi$ \\
\midrule
Memoryless~\cite{din2026distributionfreebudgetedstealthyattack}
    & $0.350$ & $1.000$ \\
One-step mag.\ predictor
    & $0.349$ & $1.000$ \\
Magnitude difference
    & $0.282$ & $0.391$ \\
\addlinespace[2pt]
Windowed energy, $L=50$
    & $0.049$ & $0.022$ \\
Fixed hold, $L=50$
    & $0.051$ & $0.011$ \\
\addlinespace[2pt]
Corrected \eqref{eq:hyst}
    & $0.294$ & $0.014$ \\
\addlinespace[2pt]
Blind
    & $0.020$ & $0.010$ \\
\bottomrule
\end{tabular}
\end{table}

Figure~\ref{fig:stealth}(a) exhibits Rem.~\ref{rem:dither} at the smallest
budget. Randomizing the score moves the location of the firing set, so both
$\pi$ and $\Psi$ fall with the scale $\sigma_{d}$. The monitor returns to
its nominal level only at $\sigma_{d}=4\sqrt S$, and there the energy
capture is $0.029$ against the blind value $0.020$: the randomized schedule
becomes stealthy by becoming blind, having discarded $92\%$ of the damage it
started with. The corrected rule reaches the same exposure, $\pi=0.014$, at
$\Psi=0.294$, an order of magnitude more damage at matched stealth. This is
Prop.~\ref{prop:decouple} read on data --- randomization perturbs the
location, while the corrected rule rearranges the run structure and leaves
the location alone.

\subsection{Budget Accuracy and the Price of Memory}

The start threshold is a plug-in order statistic, so the realized budget
inherits the accuracy of the calibration record.
Figure~\ref{fig:stealth}(b) shows the effective budget
$\hat{\bar\Gamma}$ converging to the target with a spread that contracts as
$N^{-1/2}$, and shows what depth costs: at $N=250$ the $L=50$ rule fires at
$0.177$ against a target of $0.02$, while $L=4$ fires at $0.032$, and the
three depths agree to within $0.003$ only by $N=2\times10^{4}$. The
mechanism is Prop.~\ref{prop:neff}: calibration sees only
$\alpha\Neff\approx N\Gbar/L$ exceedances, so a window of depth $L$
consumes the record by the factor $L$. The pathwise bound
$\hat\Gamma\le L\alpha$ of Prop.~\ref{prop:startrate} holds in every draw
regardless.

\begin{figure}[t]\centering
\includegraphics[width=\columnwidth]{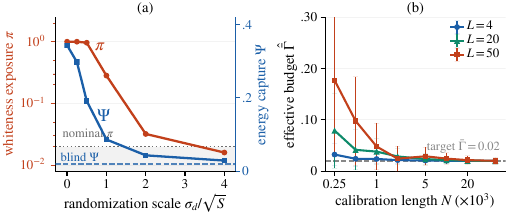}
\caption{(a) Randomization at $\Gbar=0.02$ (N2, $W=5000$)
. Shaded is the nominal level $\pi\le0.02$. (b) Effective budget
$\hat{\bar\Gamma}$ against calibration length (N2, $\Gbar=0.02$,
$\fcont=0$); markers are means over calibration draws, bars one standard
deviation, and the dashed line the target $\Gbar=0.02$.}
\label{fig:stealth}
\end{figure}

\subsection{Real Vehicle Data}

We further validate the approach on a heavy-duty-truck J1939 CAN record
\cite{biggs2024modeling} on a $10$\,Hz grid. The deployment segments are
far shorter than the reporting window $W$, so $\pi$ is not estimable here;
we report instead the lag-one exposure $\hat\rlag_{1}$, the functional the
whiteness monitor integrates, which is estimable from a single segment.
Table~\ref{tab:real} reports the result: the memoryless schedule induces
$|\hat\rlag_{1}|$ between $0.027$ and $0.065$, and the corrected rule
reduces it by factors of $1.3$ to $4.4$ at every budget while retaining
$88$--$100\%$ of the energy capture.

\begin{table}[t]\centering
\caption{Exposure on real truck data.}
\label{tab:real}
\setlength{\tabcolsep}{3pt}\footnotesize
\begin{tabular}{@{}lrrrrr@{}}
\toprule
& \multicolumn{2}{c}{memoryless \cite{din2026distributionfreebudgetedstealthyattack}}
& \multicolumn{3}{c}{corrected \eqref{eq:hyst}} \\
\cmidrule(lr){2-3}\cmidrule(lr){4-6}
$\Gbar$ & $\Psi$ & $|\hat\rlag_{1}|$ & $(L,\fcont)$ & $\Psi$ & $|\hat\rlag_{1}|$ \\
\midrule
$0.02$ & $0.424$ & $0.027$ & $(2,0.15)$ & $0.398$ & $0.013$ \\
$0.05$ & $0.556$ & $0.029$ & $(4,0.30)$ & $0.519$ & $0.012$ \\
$0.10$ & $0.681$ & $0.053$ & $(6,0.15)$ & $0.597$ & $0.012$ \\
$0.20$ & $0.811$ & $0.065$ & $(6,0.15)$ & $0.724$ & $0.020$ \\
$0.30$ & $0.886$ & $0.059$ & $(4,0.15)$ & $0.816$ & $0.024$ \\
$0.50$ & $0.964$ & $0.028$ & $(2,0.70)$ & $0.962$ & $0.022$ \\
\bottomrule
\end{tabular}
\end{table}

\section{Conclusion}
\label{sec:conclusion}

We posed budgeted sign-flip scheduling under a model-free constraint on
the autocovariance of the transmitted innovation. Damage and exposure
are separate functionals of the schedule, one fixed by where the firing
set sits in magnitude space and the other by how it is arranged in time,
so the constrained optimum thresholds a corrected score and degenerates
to the memoryless rule precisely under conditional sign symmetry. The
resulting rule keeps the pathwise magnitude-stealth certificate and pays
for memory in calibration rather than in the guarantee: the run-start
rate keeps its distribution-free bound at every window length while the
realized rate is a plug-in quantity of accuracy $\Neff\approx N/L$,
which is what bounds the usable window.


\renewcommand{\bibfont}{\scriptsize}
\printbibliography

\end{document}